\documentclass[conference]{IEEEtran}
\IEEEoverridecommandlockouts

\usepackage{amsmath,amssymb,amsfonts}
\usepackage{tabularx,booktabs}
\usepackage{mathtools}
\usepackage{algorithmic}
\usepackage{graphicx}
\usepackage{textcomp}
\usepackage{xcolor}

\usepackage{nccmath}
\def\BibTeX{{\rm B\kern-.05em{\sc i\kern-.025em b}\kern-.08em
    T\kern-.1667em\lower.7ex\hbox{E}\kern-.125emX}}
\newcolumntype{C}{>{\centering\arraybackslash}X} 
\newcolumntype{b}{>{\hsize=2.3\hsize}X}
\usepackage{amsthm}
\usepackage{bbm}
\usepackage{csquotes}
\usepackage{chngcntr}
\usepackage{apptools}

\usepackage{cite}

\theoremstyle{plain}
\newtheorem{theorem}{Theorem}
\newtheorem{lemma}{Lemma}
\newtheorem{corollary}{Corollary}

\theoremstyle{definition}
\newtheorem{dfn}{Definition}

\theoremstyle{remark}

\newtheorem{remark}{Remark}

\newcommand{\X}{\mathcal{X}}
\newcommand{\Y}{\mathcal{Y}}
\newcommand{\Z}{\mathcal{Z}}

\newcommand{\Pm}{\mathcal{P}}
\newcommand{\Q}{\mathcal{Q}}
\newcommand{\F}{\mathcal{F}}

\makeatletter
\newcounter{labelcnt}
\renewcommand{\thelabelcnt}{(\alph{labelcnt})}
\newcommand{\setlabel}[1]{%
  \refstepcounter{labelcnt}\ltx@label{lbl:#1}%
  {\text{\upshape\thelabelcnt}}%
}

\newcommand{\reflabel}[1]{\text{\upshape\ref{lbl:#1}}}
\makeatother
\DeclareMathOperator*{\esssup}{ess\,sup}

\newcommand{\ml}[2]{\mathcal{L}\left(#1  \!\!  \to  \!\!   #2\right)} 

\begin{document}

\title{On conditional Sibson's $\alpha$-Mutual Information}

\author{
\IEEEauthorblockN{Amedeo Roberto Esposito, Diyuan Wu, Michael Gastpar}
\IEEEauthorblockA{\textit{School of Computer and Communication Sciences} \\
       EPFL, Lausanne, Switzerland\\
\{amedeo.esposito, diyuan.wu, michael.gastpar\}@epfl.ch}
}

\maketitle

\begin{abstract}
   In this work, we analyse how to define a conditional version of Sibson's $\alpha$-Mutual Information. Several such definitions can be advanced and they all lead to different information measures with different (but similar) operational meanings. We will analyse in detail one such definition, compute a closed-form expression for it and endorse it with an operational meaning while also considering some applications. The alternative definitions will also be mentioned and compared.
\end{abstract}


\section{Introduction}
Sibson's $\alpha$-Mutual Information is a generalization of Shannon's Mutual Information with several applications in probability, information and learning theory \cite{fullVersionGeneralization}. In particular, it has been used to provide concentration inequalities in settings where the random variables are \textbf{not} independent, with applications to learning theory \cite{fullVersionGeneralization}. The measure is also connected to Gallager's exponent function, a central object in the channel coding problem both for rates below and above capacity \cite{gallager1,gallager2}. Moreover, a new operational meaning has been given to the measure with $\alpha\!=\!+\infty$ when a novel measure of information leakage has been proposed in \cite{leakageLong}, under the name of Maximal leakage. 
Similarly to $I_\alpha$, Maximal Leakage has recently found applications in learning and probability theory \cite{fullVersionGeneralization}. Howerever, while Maximal Leakage has a corresponding conditional form \cite{leakageLong}, Sibson's $\alpha$-Mutual Information lacks an agreed upon conditional version. In this work we analyse a path that could be taken in defining such a measure and will focus on one specific choice, given in Definition \ref{conditionalSibs} below. We discuss key properties of this choice and endow it with an operational meaning as the error-exponent in a properly defined hypothesis testing problem. Moreover, we hint at some application of this measure to other settings as well. The choice we make is not unique and we will explain how making different choices leads to different information measures, all of them equally meaningful. A conditional version of Sibson's $I_\alpha$ has been presented in~\cite{tomamichel}. We briefly present their measure in Sec. \ref{sec-otherCond} along with a new result that we believe to be of interest. We then present in Sec. \ref{sec-newCond} a different choice for conditional $I_\alpha$. We show some properties of this measure, compare the two objects in Sec. \ref{sec-discussion} and then discuss a general approach to associate an operational meaning to these measures in Sec. \ref{sec-opMeaning}.
Alternative routes have been considered in~\cite{alphaLeakage} where Arimoto's generalisation of the Mutual Information has been considered and a conditional version has been given.
\section{Background and definitions}
Given a function $f:\mathbb{R}\to [-\infty,+\infty]$ we can define its convex conjugate $f^\star:\mathbb{R}\to [-\infty,+\infty]$ as follows:
\begin{equation}
    f^\star(\lambda) = \sup_x(\lambda x-f(x)).
\end{equation}
Given a function $f$, $f^\star$ is guaranteed to be lower semicontinuous and convex. We can re-apply the conjugation operator to $f^\star$ and obtain $f^{\star\star}$. If $f$ is convex and lower semincontinuous then $f=f^{\star\star}$, otherwise all we can say is that $\forall x\in \mathbb{R}\,\,f^{\star\star}(x)\leq f(x).$  $\log$ denotes the natural logarithm.
\subsection{Sibson's $\alpha$-Mutual Information}
Introduced by R\'enyi as a generalization of entropy and KL-divergence, $\alpha$-divergence has found many applications ranging from hypothesis testing to guessing and several other statistical inference and coding problems~\cite{verduAlpha}. Indeed, it has several useful operational interpretations (e.g., hypothesis testing, and the cut-off rate in block coding \cite{RenyiKLDiv,opMeanRDiv1}). It can be defined as follows~\cite{RenyiKLDiv}.
\begin{dfn}
	Let $(\Omega,\F,\Pm),(\Omega,\F,\Q)$ be two probability spaces. Let $\alpha>0$ be a positive real number different from $1$. Consider a measure $\mu$ such that $\Pm\ll\mu$ and $\Q\ll\mu$ (such a measure always exists, e.g. $\mu=(\Pm+\Q)/2$)) and denote with $p,q$ the densities of $\Pm,\Q$ with respect to $\mu$. The $\alpha$-Divergence of $\Pm$ from $\Q$ is defined as follows:
	\begin{align}
	D_\alpha(\Pm\|\Q)=\frac{1}{\alpha-1} \log \int p^\alpha q^{1-\alpha} d\mu.
	\end{align}
\end{dfn}
\begin{remark}
    The definition is independent of the chosen measure $\mu$. 
    It is indeed possible to show that
    $\int p^{\alpha}q^{1-\alpha} d\mu = \int \left(\frac{q}{p}\right)^{1-\alpha}d\Pm $, and that whenever $\Pm\ll\Q$ or $0<\alpha<1,$ we have $\int p^{\alpha}q^{1-\alpha} d\mu= \int \left(\frac{p}{q}\right)^{\alpha}d\Q$, see \cite{RenyiKLDiv}.
\end{remark}

It can be shown that if $\alpha>1$ and $\Pm\not\ll\Q$ then $D_\alpha(\Pm\|\Q)=\infty$. The behaviour of the measure for $\alpha\in\{0,1,\infty\}$ can be defined by continuity. In general, one has that $D_1(\Pm\|\Q) = D(\Pm\|\Q)$ but if $D(\Pm\|\Q)=\infty$ or there exists $\beta$ such that $D_\beta(\Pm\|\Q)<\infty$ then $\lim_{\alpha\downarrow1}D_\alpha(\Pm\|Q)=D(\Pm\|\Q)$\cite[Theorem 5]{RenyiKLDiv}. For an extensive treatment of $\alpha$-divergences and their properties we refer the reader to~\cite{RenyiKLDiv}. 
Starting from R\'enyi's Divergence and the geometric averaging that it involves, Sibson built the notion of Information Radius \cite{infoRadius}:
\begin{dfn}\label{SibsonsInfoRadius}
    Let $(\mu_1,\ldots,\mu_n)$ be a family of probability measures and $(w_1,\ldots, w_n)$ be a set of weights s.t. $w_i\geq 0$ for $i=1,\ldots,n$ and such that $\sum_{i=1}^n w_i>0$. Let $\alpha\geq1$, the information radius of order $\alpha$ is defined as:
    \begin{align*}
        \frac{1}{\alpha-1}\min_{\nu\ll\sum_iw_i\mu_i}\log\left(\sum_i w_i\exp((\alpha-1)D_\alpha(\mu_i\|\nu))\right)  \label{infoRadius}.
    \end{align*}
\end{dfn}
Suppose now we have two random variables $X,Y$ jointly distributed according to $\Pm_{XY}$. It is possible to generalise Def. \ref{SibsonsInfoRadius} and see that the information radius is a special case of the following quantity \cite{verduAlpha}:
\begin{equation}
    I_\alpha(X,Y) = \min_{\Q_Y} D_\alpha(\Pm_{XY}\|\Pm_{X}\Q_Y). \label{sibsIAlpha} 
\end{equation}
$I_\alpha(X,Y)$ represents a generalisation of Shannon's Mutual Information and possesses many interesting properties \cite{verduAlpha}. Indeed, $\lim_{\alpha\to 1}I_\alpha(X,Y)=I(X;Y)$. 
On the other hand when $\alpha\to\infty$, we get: $I_\infty(X,Y)=\log\mathbb{E}_{\Pm_Y}\left[\sup_{x:\Pm_X(x)>0} \frac{\Pm_{XY}(\{x,Y\})}{\Pm_X(\{x\})\Pm_Y(\{Y\})}\right] =\ml{X}{Y}$, where $\ml{X}{Y}$ denotes the Maximal Leakage from $X$ to $Y$, a recently defined information measure with an operational meaning in the context of privacy and security \cite{leakageLong}.
For more details on Sibson's $\alpha$-MI, as well as a closed-form expression, we refer the reader to~\cite{verduAlpha}, as for Maximal Leakage the reader is referred to~\cite{leakageLong}.
\section{Definition}
\subsection{Introduction}\label{Sec-Def-subsec-intro}
The characterisation expressed in \eqref{sibsIAlpha} represents the foundation of this work. Indeed, using \eqref{sibsIAlpha} as the definition of Sibson's $\alpha$-MI allows us to draw parallels with Shannon's Mutual Information. This, in turn, allows us to define, drawing inspiration from Shannon's measures, an analogous conditional version of Sibson's $I_\alpha$. It is very well known that $I(X;Y)=D(\Pm_{XY}\|\Pm_{X}\Pm_{Y})$ as well as $I(X;Y|Z)=D(\Pm_{XYZ}\|\Pm_Z\Pm_{X|Z}\Pm_{Y|Z})$. We can thus follow a similar approach in defining a conditional $\alpha$-Mutual Information: we will estimate the (R\'enyi's) divergence of the joint $\Pm_{XYZ}$ from a distribution characterised by the Markov chain $X-Z-Y$ via $\alpha$-Divergences. Mimicking \eqref{sibsIAlpha} we will also minimise such divergence with respect to a family of measures. Having now three random variables, we can think of three natural factorisations for $\Pm_{XYZ}$ (assuming that $X-Z-Y$ holds): $\Pm_X\Pm_{Z|X}\Pm_{Y|Z}$, $\Pm_Y\Pm_{Z|Y}\Pm_{X|Z}$, $\Pm_Z\Pm_{Y|Z}\Pm_{X|Z}$. The question then is: which measure should we minimise with respect to, in order to define $I_\alpha(X,Y|Z)$? Natural candidates seem to be the minimisations with respect to $\Q_Z$,$\Q_{Y|Z}$ and $\Q_Y$.
The matter is strongly connected to the operational meaning that the information measure acquires, alongside with the applications it can provide. Each of the definitions can be useful in specific settings. Keeping this in mind, the purpose of this work is not to compare different definitions in order to find the best one but rather to highlight properties of the different definitions with an operationally driven approach. Each of these measures can be associated to a hypothesis testing problem and a bound relating different measures of the same event (typically a joint and a Markov chain-like distribution). Different applications require different conditional $I_\alpha$'s. 
With this drive, let us make a specific choice for the minimisation and draw a parallel with the others along the way. The random variable whose measure\footnote{It is clearly possible to minimise over more than one random variable at once, like it has been done in \cite{lapidothTesting,tomamichel} in the context of both regular $I_\alpha(X,Y)$ and conditional $I_\alpha(X,Y|Z)$. } we choose to minimise will be denoted as a superscript.
\subsection{$I^{Y|Z}_\alpha(X,Y|Z)$}\label{sec-otherCond}
In~\cite{tomamichel}, conditional $\alpha$-mutual information was defined as follows:
\begin{dfn}\label{otherCondIAlpha}
   Let $X,Y,Z$ be three random variables jointly distributed according to $\Pm_{XYZ}$. For $\alpha>0$, a conditional Sibson's mutual information of order $\alpha$ between $X$ and $Y$ given $Z$ is defined as:
\begin{equation}
    I^{Y|Z}_\alpha(X,Y|Z) = \min_{\Q_{Y|Z}} D_\alpha(\Pm_{XYZ}\|\Pm_{X|Z}\Q_{Y|Z}\Pm_Z).
\end{equation}
\end{dfn}
It is possible to find a closed-form expression for Def. \ref{otherCondIAlpha} \cite[Section IV.C.2]{tomamichel}.
This definition is interesting as setting $Z$ equal to a constant allows us to retrieve $I_\alpha(X,Y)$. Moreover, 
starting from Definition \ref{otherCondIAlpha} and its closed-form expression one can retrieve the following result.
\begin{theorem}\label{probBound2}
Let $(\X\times\Y\times \Z,\F,\Pm_{XYZ})$ be a probability space.
Let $\Pm_Z$ and $\Pm_{X|Z}$ be the induced conditional and marginal distributions.
Assume that $\Pm_{XYZ}\ll\Pm_Z\Pm_{Y|Z}\Pm_{X|Z}.$
Given $E\in \F$ and $z\in\Z, y\in\Y$, let $E_{z,y} = \{x : (x,y,z)\in E\}$. Then, fixed $\alpha\geq1$:
\begin{align}
    \Pm_{XYZ}(E) \leq &\mathbb{E}_{Z}\left[\esssup_{\Pm_{Y|Z}}\Pm_{X|Z}(E_{Z,Y})\right]^{\frac{\alpha-1}{\alpha}} \notag \\  & \cdot\exp\left(\frac{\alpha-1}{\alpha}I_\alpha^{Y|Z}(X,Y|Z)\right).\label{probBoundEq2}
\end{align}
\end{theorem}
\begin{proof}
\begin{align}
    \Pm_{XYZ}(E) =\, &\mathbb{E}_{\Pm_Z\Pm_{Y|Z}\Pm_{X|Z}}\left[\frac{d\Pm_{XYZ}}{d\Pm_Z\Pm_{Y|Z}\Pm_{X|Z}} \mathbbm{1}_E\right] \\
    \leq\, &\mathbb{E}_{\Pm_Z}^{\frac{1}{\alpha''}}\left[\mathbb{E}_{\Pm_{Y|Z}}^\frac{\alpha''}{\alpha'}\left[\mathbb{E}_{\Pm_{X|Z}}^{\frac{\alpha'}{\alpha}}\left[ \left(\frac{d\Pm_{XYZ}}{d\Pm_Z\Pm_{Y|Z}\Pm_{X|Z}}\right)^{\alpha}\right]\right]\right] \notag \\
    \cdot\,  &\mathbb{E}_{\Pm_Z}^\frac{1}{\gamma''}\left[\mathbb{E}_{\Pm_{Y|Z}}^\frac{\gamma''}{\gamma'}\left[\mathbb{E}_{\Pm_{X|Z}}^{\frac{\gamma'}{\gamma}}\left[\mathbbm{1}_E^{\gamma}\right]\right]\right] \\
    \leq\, &\mathbb{E}_{Z}\left[\esssup_{\Pm_{Y|Z}}\Pm_{X|Z}(E_{Z,Y})\right]^{\frac{\alpha-1}{\alpha}} \notag \\  & \cdot \exp\left(\frac{\alpha-1}{\alpha}I_\alpha^{Y|Z}(X,Y|Z)\right). \label{statementThm1}
\end{align}
The first inequality follows from applying H\"older's inequality three times and the six parameters are such that $\frac{1}{\alpha''}+\frac{1}{\gamma''}=\frac{1}{\alpha'}+\frac{1}{\gamma'}=\frac{1}{\alpha}+\frac{1}{\gamma}=1.$ \eqref{statementThm1} follows from setting $\alpha''=\alpha$ and $\alpha'=1$ which imply $\gamma''=\gamma$ and  $\gamma'\to\infty$.
\end{proof}
Another property of $I_\alpha^Z$ is that, similarly to unconditional $I_\alpha$ \cite{leakageLong}, taking the limit of $\alpha\to\infty$, we have that $I_\alpha^{Y|Z}(X,Y|Z)\!\!\underset{\alpha\to\infty}{\to}\!\! \mathcal{L}(X\!\!\to\!\!Y|Z),$
leading us to the following:
\begin{corollary}
Under the same assumptions of Theorem \ref{probBound2}:
\begin{align}
   \Pm_{XYZ}(E) \leq &\mathbb{E}_{Z}[\esssup_{\Pm_{Y|Z}}\Pm_{X|Z}(E_{Z,Y})] \exp\left(\mathcal{L}(X\!\! \to \!\!Y|Z)\right) \label{probBoundLeakageEq}.
\end{align}
\end{corollary}
\subsection{$I^Z_\alpha(X,Y|Z)$}\label{sec-newCond}
As discussed in Section~\ref{Sec-Def-subsec-intro},
another natural candidate definition of conditional $\alpha$-mutual information is the following:
\begin{dfn}\label{conditionalSibs}
    Under the same assumptions of Definition \ref{otherCondIAlpha}:
    \begin{align}
        I^Z_\alpha(X,Y|Z) = \min_{\Q_Z} D_\alpha(\Pm_{XYZ}\|\Pm_{X|Z}\Pm_{Y|Z}\Q_Z) \label{condIAlphaDef}.
    \end{align}
\end{dfn}
To the best of our knowledge Definition \ref{conditionalSibs} has not been considered elsewhere.
As for $I_\alpha^{Y|Z}(X,Y|Z)$, it is possible to compute a closed-form  expression for $I_\alpha^Z(X,Y|Z)$. We will limit ourselves to discrete random variables for simplicity.
\begin{theorem}
Let $\alpha>0$ and $X,Y,Z$ be three discrete random variables.
\begin{align*}
    &I^Z_\alpha(X,Y|Z)= \frac{\alpha}{\alpha-1}\log\sum_z \Pm_Z(z) \notag \\ &\cdot\left(\sum_{x,y}  \Pm_{XY|Z=z}(x,y)^\alpha(\Pm_{X|Z=z}(x)\Pm_{Y|Z=z}(y))^{1-\alpha}\right)^\frac{1}{\alpha}.
\end{align*}
\end{theorem}
The proof follows from the definition of $I^Z_\alpha(X,Y|Z)$ and Sibson's identity \cite[Eq. (12)]{opMeanRDiv1}.
Mirroring Section \ref{sec-otherCond} we can state an analogous of Theorem~\ref{probBound2} for $I^Z_\alpha$:
\begin{theorem}\label{probBound}
Let $(\X\times\Y\times \Z,\F,\Pm_{XYZ})$ be a probability space.
Let $\Pm_{Y|Z}$ and $\Pm_{X|Z}$ be the induced conditional distributions.
Assume that $\Pm_{XYZ}\ll\Pm_Z\Pm_{Y|Z}\Pm_{X|Z}.$
Given $E\in \F$ and $z\in\Z$, let $E_z = \{(x,y) : (x,y,z)\in E\}$. Then, fixed $\alpha\geq1$:
\begin{align}
    \Pm_{XYZ}(E) \leq &\esssup_{\Pm_Z}\left(\Pm_{X|Z}\Pm_{Y|Z}(E_Z)\right)^{\frac{\alpha-1}{\alpha}}\notag\\ & \cdot \exp\left(\frac{\alpha-1}{\alpha}I_\alpha^Z(X,Y|Z)\right).\label{probBoundEq}
\end{align}
\end{theorem}

This type of result is useful as it allows us to approximate the probability of $E$ under a joint, with the probability of $E$ under a different measure encoding some independence (typically easier to analyse) --- in this specific case, the measure induced by a Markov chain. Such bounds represent, for us, the main application-oriented employment of these measures \cite{fullVersionGeneralization}.
Notice that, other than using $I_\alpha^Z$ instead of $I_\alpha^{Y|Z}$, Theorem~\ref{probBound} involves a different essential supremum as compared to Theorem~\ref{probBound2}.
Moving on with the comparison, we have that differently from Definition~\ref{otherCondIAlpha}, the information measure we are defining here is symmetric. Moreover, setting $Z$ to a constant in Definition~\ref{conditionalSibs} does not allow us to retrieve $I_\alpha(X,Y)$, but rather $D_\alpha(\Pm_{XY}\|\Pm_X\Pm_Y)$. 
\subsection{An additive SDPI-like inequality}
Definition \ref{conditionalSibs} shares some interesting properties with $I_\alpha(X,Y)$. One such property is a rewriting of $I_\alpha(X,Y)$ in terms of $D_\alpha$. This allows us to leverage the strong data processing inequality (SDPI) for Hellinger integrals of order $\alpha$, which in turn allows us to provide an SDPI-like results for $I_\alpha^{Z}$. A definition for SDPIs can be found at \cite[Def 3.1]{sdpiRaginsky}

More precisely, we can write 
\begin{align}
   &I^Z_\alpha(X,Y|Z) \notag \\ &=\frac{\alpha}{\alpha-1}\log\mathbb{E}_Z\!\left[\exp\!\left(\!\frac{\alpha-1}{\alpha}D_\alpha(\Pm_{XY|Z}\|\Pm_{X|Z}\Pm_{Y|Z})\!\right)\!\right]\notag\\
    &=\frac{\alpha}{\alpha-1}\log\mathbb{E}_Z\!\!\left[\left(D_{f_\alpha}(\Pm_{XY|Z}\|\Pm_{X|Z}\Pm_{Y|Z})\right)^{1/\alpha}\right], \label{lmgfRepresentation}
\end{align}
where $D_{f_\alpha}$ denotes the Hellinger integral of order $\alpha$, \textit{i.e.,} given two measures $\Pm,\Q$, $D_{f_\alpha}(\Pm\|\Q)=\mathbb{E}_{\Q}\left[\left(\frac{d\Pm}{d\Q}\right)^{\alpha}\right]$.
Leveraging Eq. \eqref{lmgfRepresentation} we can state the following.
\begin{theorem}
\label{sdpiCondIAlpha}
Let $\alpha>1$ and $X,Y,W,Z$ be four random variables such that $(Z,W)-X-Y$ is a Markov chain:
\begin{equation}
    I^Z_\alpha(W,Y|Z) \leq \frac{1}{\alpha-1}\log\left(\eta_{f_\alpha}(\Pm_{Y|X})\right) + I^Z_\alpha(W,X|Z),
\end{equation}
where we denote with $\eta_{f_\alpha}(\Pm_{Y|X})$ the contraction parameter of the Hellinger integral of order $\alpha$, i.e., for a given Markov Kernel $K$, $\eta_{f_\alpha}(K) = \sup_{\mu, \nu\neq \mu} \frac{D_{f_\alpha}(K\mu\|K\nu)}{D_{f_\alpha}(\mu\|\nu)}$ \cite[Def. III.1]{sdpiRaginsky}.
\end{theorem}
The proof follows from Eq. \eqref{lmgfRepresentation} and a reasoning similar to \cite[Lemma 3]{distrFuncComput} but applied to the $D_{f_\alpha}$-divergence instead of the KL-divergence.
\begin{remark}
Notice that data processing inequalities are simply a consequence of the convexity of $f$\cite[Thm 4.2]{lectureNotesSDPI} and $f_\alpha(x) = x^\alpha$ is indeed convex. Hence, although the Hellinger integral is not normalised to be $0$ whenever the measures are the same, it does satisfy a DPI. 
 Moreover, the contraction parameter of a strong data-processing inequality is always less than or equal to $1$. Hence, $\log(\eta_{f_\alpha}(K)) \leq 0$.
\end{remark}

An analogous result of Theorem~\ref{sdpiCondIAlpha} for Definition~\ref{otherCondIAlpha} does not seem possible.
\begin{remark}
One can state a result similar to Theorem \ref{sdpiCondIAlpha} for unconditional $I_\alpha$. Specifically, we can write
$$I_\alpha(X,Y) = \frac{\alpha}{\alpha-1}\log\mathbb{E}_Y\left[D_{f_\alpha}^{1/\alpha}(\Pm_{X|Y}\|\Pm_X)\right].$$
Since $I_\alpha$ is an asymmetric quantity, we only get the SDPI-like result in one direction.
Namely, given the Markov chain
   $W-X-Y$ we can relate via SDPI $I_\alpha(W,Y)$ and $I_\alpha(X,Y)$ (but, for instance, not $I_\alpha(W,X)$ and $I_\alpha(X,Y)$), as follows:
    \begin{equation}
    I_\alpha(W,Y) \leq \frac{1}{\alpha-1}\log\left(\eta_{f_\alpha}(\Pm_{W|X})\right) + I_\alpha(X,Y). \label{sdpiIAlpha}
\end{equation}
\end{remark}

Theorem~\ref{sdpiCondIAlpha} and Eq. \eqref{sdpiIAlpha} represent a different from usual SDPI-like inequality. The reason for this is that the (function of the) $\eta$ parameter is added to the information measure, rather than multiplied. However, one of the main applications of (conditional and not) $I_\alpha$ in bounds requires the exponentiation of the quantity, which brings us back to a multiplicative form. To make this statement more precise, let us state the following:
\begin{corollary}\label{tighterProbBound}
Under the same assumptions of Theorem~\ref{sdpiCondIAlpha} we have that:
\begin{align*}
    \Pm_{WYZ}(E) \leq &\esssup_{\Pm_Z}(\Pm_{W|Z}\Pm_{Y|Z}(E_Z))^{\frac{\alpha-1}{\alpha}} \notag \\ &\cdot\left(\eta_{f_\alpha}(\Pm_{Y|X})\right)^{1/\alpha}\cdot\exp\left(\frac{\alpha-1}{\alpha}I_\alpha^Z(W,X|Z)\right).
\end{align*}
\end{corollary}
Corollary~\ref{tighterProbBound} follows directly from Theorem~\ref{probBound} and Theorem~\ref{sdpiCondIAlpha}.
\begin{remark}A similar result can be derived for unconditional $I_\alpha$ starting from \eqref{sdpiIAlpha} and \cite[Corollary 1]{fullVersionGeneralization}.\end{remark} 

\subsection{Discussion on $I_\alpha^Z$ and $I_\alpha^{Y|Z}$}\label{sec-discussion}
Let us now use Theorems \ref{probBound2} and \ref{probBound} as a means of comparison for the two conditional $I_\alpha$. These results are useful whenever we want to control the joint measure of some event $E$ but we only know how to control it (e.g., via an upper-bound) under some hypothesis of independence \cite{fullVersionGeneralization}. Consider the factorisation of $\Pm_{XYZ}$ under $X-Z-Y$ to be fixed. In the context of Theorem \ref{probBound2} and \ref{probBound},
according to the measure we know how to control, different conditional $I_\alpha$'s will appear on the right-hand side of the bound (c.f., Eq. \eqref{probBoundEq2}, \eqref{probBoundLeakageEq} and \eqref{probBoundEq}). For instance, if we assume to be able to control $\esssup_{\Q_Z}\left(\Pm_{X|Z}\Pm_{Y|Z}(E_Z)\right)$ then, Theorem \ref{probBound} tells us that $I_\alpha^Z(X,Y|Z)$ is the measure to study. If we assume instead that we are able to control terms of the form $\mathbb{E}_{P_Z}[\esssup_{\Pm_{Y|Z}}\Pm_{X|Z}(E_{Z,Y})]$ then $I_\alpha^{Y|Z}(X,Y|Z)$ would be the measure to analyse. (Quantities like $\mathbb{E}_{P_Z}[\esssup_{\Pm_{Y|Z}}\Pm_{X|Z}(E_{Z,Y})]$, for specific choices of $E$, are known in the literature as \enquote{small-ball probabilities} and have found applications in distributed estimation problems and distributed function computation \cite{bayesRiskRaginsky,distrFuncComput}).
More generally, we can find a duality between the measure over which we supremise (on the right-hand side of the bounds) and the corresponding minimisation in the definition of conditional $I_\alpha$.
The same measures also have a fundamental role in defining the hypothesis testing problem that endows the information measure with its operational meaning, as we will see in the next section. 
\section{Operational Meaning}\label{sec-opMeaning}
Drawing inspiration from \cite{tomamichel, lapidothTesting, lapidothTestingJournal}, let us consider the following composite hypothesis testing problem. Fix a pmf $\Pm_{XYZ}$, observing a sequence of triples  $\{(X_i,Y_i,Z_i)\}_{i=1}^n$ we want to decide whether:
\begin{enumerate}
    \setcounter{enumi}{-1}
    \item $\{(X_i,Y_i,Z_i)\}_{i=1}^n$ is sampled in an iid fashion from $\Pm_{XYZ}$ (null hypothesis);
    \item $\{(X_i,Y_i,Z_i)\}_{i=1}^n$ is sampled in an iid fashion from $\Q_Z\Pm_{X|Z}\Pm_{Y|Z}$, where $\Q_Z$ is an arbitrary pmf over the space $\Z$ (alternative hypothesis).
 \end{enumerate}
 We can relate $I_\alpha^Z(X,Y|Z)$ to the error-exponent of the just defined hypothesis testing problem. This can be seen as a more lenient test for markovity where the measure of $Z$ is allowed to vary. Similarly to before, there is a link between which measure is allowed to vary and the minimisation in the definition of conditional $I_\alpha$. Choosing, for instance, to minimise over 
 $\Q_X$ allows this measure to vary in the alternative hypothesis. Using Theorem \ref{probBound} we can already connect $I_\alpha^Z$ to the problem in question. Given a test $T_n : \{\X\times\Y\times\Z\}^n \to \{0,1\}$, we will denote with $p_n^1$ (Type-1 error) the probability of wrongfully choosing the hypothesis $1$ given that the sequence is distributed according to $\Pm_{XYZ}^{\otimes n}$, i.e. $p_n^1 = \Pm_{XYZ}^{\otimes n}(T_n(\{(X_i,Y_i,Z_i)\}_{i=1}^n) = 1)$ and with $p_n^2$ (Type-2 error) the maximum probability of wrongfully choosing the hypothesis $0$ given that the sequence is distributed according to $(\Q_Z\Pm_{X|Z}\Pm_{Y|Z})^{\otimes n}$ for some $\Q_Z$, i.e. $p_n^2 = \sup_{\Q_Z\in \mathcal{P}(\Z)} (\Q_Z\Pm_{X|Z}\Pm_{Y|Z})^{\otimes n}(T_n(\{(X_i,Y_i,Z_i)\}_{i=1}^n) = 0)$.
 \begin{theorem}\label{type1ErroBound} Let $n>0$ and $T_n : \{\X\times\Y\times\Z\}^n \to \{0,1\}$ be a deterministic test, that upon observing the sequence $\{(X_i,Y_i,Z_i)\}_{i=1}^n$ chooses either the null or the alternative hypothesis.
 Assume that $\exists R>0:$ $\forall \Q_Z \in \mathcal{Q}(\Z)$ we have $ (\Q_Z\Pm_{X|Z}\Pm_{Y|Z})^{\otimes n}(T_n(\{(X_i,Y_i,Z_i)\}_{i=1}^n) = 0) \leq \exp(-nR)$. Let also $\alpha\geq1$,
 \begin{align}
     1- p^1_n \leq \exp\left(-\frac{\alpha-1}{\alpha}n(R-I_\alpha^Z(X,Y|Z))\right).
 \end{align}
 \end{theorem} 
 \begin{proof}
 We have that $1-p_n^1 = \Pm_{XYZ}^{\otimes n}(T_n(\{(X_i,Y_i,Z_i)\}_{i=1}^n) = 0)$. Starting from Theorem \ref{probBound}:
 \begin{align}
     1-p_n^1 \leq &\esssup_{\Pm^n_Z}\left(\Pm^n_{X|Z}\Pm^n_{Y|Z}(E^n_Z)\right)^{1/\gamma} \notag\\ & \cdot\exp\left(\frac{\alpha-1}{\alpha}I_\alpha^Z(X^n,Y^n|Z^n)\right)\label{thmErrorExp}.
 \end{align}
 Since we assumed the exponential decay of $(\Q_Z\Pm_{X|Z}\Pm_{Y|Z})^{\otimes n}(T_n(\{(X_i,Y_i,Z_i)\}_{i=1}^n) = 0)$ for every $\Q_Z$ we also have that $\esssup_{\Pm^n_Z}\left(\Pm^n_{X|Z}\Pm^n_{Y|Z}(E^n_Z)\right) \leq \exp(-nR)$ (consider a measure $\tilde{Q}_Z$ that puts all the mass on the sequence achieving the essential supremum in \eqref{thmErrorExp}). Given the assumption of independence on the triples $\{(X_i,Y_i,Z_i)\}_{i=1}^n$ and following a reasoning similar to the one in Eqn. (49) in \cite{verduAlpha}, we have that $I_\alpha^Z(X^n,Y^n|Z^n)=nI_\alpha^Z(X,Y|Z)$. The conclusions then follow from algebraic manipulations of \eqref{thmErrorExp}.
 \end{proof}
 This result implies that if we assume an exponential decay for the type-2 error $p_n^2$ and $R>I_\alpha^Z(X,Y|Z)$ we have an exponential decay of the probability of correctly choosing the null hypothesis as well. Moreover, for every $n>0$:
 \begin{equation}
     \frac1n\log(1-p_n^1) \leq -\frac{\alpha-1}{\alpha}\left(R- I_\alpha^Z(X,Y|Z)\right).
 \end{equation}
 We can conclude that:
 \begin{equation}
     \limsup_{n\to\infty}\frac1n\log(1-p_n^1) \leq -\!\!\!\!\sup_{\alpha\in(1,+\infty]} \!\!\frac{\alpha-1}{\alpha}\left(R- I_\alpha^Z(X,Y|Z)\right).\notag
 \end{equation}
\subsection{Error exponents}
Following the approach undertaken in \cite{lapidothTesting} we can also define an achievable error-exponent pair for the hypothesis testing problem in question.
\begin{dfn}
A pair of error exponents $(E_P,E_Q) \in \mathbb{R}^2$ is called achievable w.r.t the above hypothesis testing problem  if there exists a series of tests $\{T_n\}_{n=1}^\infty$ such that \footnote{ As pointed out in \cite{lapidothTesting}, despite having bounds like in Theorem \ref{type1ErroBound} decaying with two rates $E_P,E_Q$, we cannot conclude anything on the achievability of the pair.}:
\begin{align*}
    &\liminf_{n \rightarrow \infty} -\frac{1}{n}\log \Pm_{XYZ}^{\otimes n}(T_n(\{(X_i,Y_i,Z_i)\}_{i=1}^n) = 1) > E_P,\\
    &\liminf_{n \rightarrow \infty} \inf_{Q_Z} -\frac{1}{n} \\ &\log (\Q_Z\Pm_{X|Z}\Pm_{Y|Z})^{\otimes n}(T_n(\{(X_i,Y_i,Z_i)\}_{i=1}^n) = 0) > E_Q.
\end{align*}
We can then define the error exponent functions \cite{lapidothTesting} $E_P:\mathbb{R}\to \mathbb{R}\cup \{+\infty\}$ and $E_Q:\mathbb{R}\to \mathbb{R}\cup \{+\infty\}$ as follows:
\begin{align}
    & E_P(E_Q) = \sup\{E_P\in \mathbb{R} : (E_P,E_Q) \text{ is achievable}\} \\
    & E_Q(E_P) = \sup\{E_Q\in \mathbb{R} : (E_P,E_Q) \text{ is achievable}\}
\end{align}
\end{dfn}
It is now possible to relate $I_\alpha^Z(X,Y|Z)$, where $\alpha \in (0,1]$, with both the Fenchel conjugate of $E_P(\cdot)$, $E_P^{\star}(\cdot)$ and $E_P^{\star\star}(\cdot)$.
 First, let us characterise $E_P^\star(E_Q)$.
\begin{lemma}\label{conjugateEP}
\begin{equation}
    E_P^\star(\lambda) = 
    \left\{
             \begin{array}{lr}
             +\infty, &\text{if } \lambda>0   \\
             \lambda I_{\frac{1}{1-\lambda}}(X,Y|Z), & \text{otherwise}.\\
             \end{array}
\right.
\end{equation}
\end{lemma}
\begin{proof}
Assume $\lambda \leq 0$,
\begin{align*}
  &E_P^\star(\lambda)=\sup_{E_Q \in \mathbb{R}} \left[\lambda E_Q - E_P(E_Q)\right]\\
        &= \underset{E_Q \in  \mathbb{R} }{\sup}\left[\lambda E_Q - \!\!\!\!\underset{\substack{\mathcal{R}_{XYZ}: \\D(\mathcal{R}_{XYZ}||R_{Z}P_{X|Z}P_{Y|Z})\leq E_Q}}{\inf} \!\!\!\!D(\mathcal{R}_{XYZ}\|\Pm_{XYZ})\right] \\
        &\overset{\setlabel{step1}}{=} \underset{E_Q \in  \mathbb{R}}{\sup} \underset{\substack{\mathcal{R}_{XYZ}:\\D(\mathcal{R}_{XYZ}\|\mathcal{R}_{Z}\Pm_{X|Z}\Pm_{Y|Z})\leq E_Q}}{\sup} \!\!\!\!\left[\lambda E_Q -  D(\mathcal{R}_{XYZ}\|\Pm_{XYZ})\right] \\
        &= \underset{\mathcal{R}_{XYZ}}{\sup}\!\! \underset{\substack{E_Q \in \mathbb{R}:\\ E_Q \geq D(\mathcal{R}_{XYZ}\|\mathcal{R}_{Z}\Pm_{X|Z}\Pm_{Y|Z})}}{\sup}\!\!\!\!\left[\lambda E_Q -  D(\mathcal{R}_{XYZ}\|\Pm_{XYZ})\right] \\
       &\overset{\setlabel{step2}}{=} \underset{\mathcal{R}_{XYZ}}{\sup} \left[\lambda D(\mathcal{R}_{XYZ}\|\mathcal{R}_{Z}\Pm_{X|Z}\Pm_{Y|Z}) -  D(\mathcal{R}_{XYZ}\|\Pm_{XYZ})\right]\\
         &\overset{\setlabel{step3}}{=} (\lambda-1) \underset{\Q_Z}{\inf} \underset{\mathcal{R}_{XYZ}}{\inf}\big[\frac{-\lambda}{1-\lambda} D(\mathcal{R}_{XYZ}\|\Q_{Z}\Pm_{X|Z}\Pm_{Y|Z})\\
         &\,+\frac{1}{1-\lambda} D(\mathcal{R}_{XYZ}\|\Pm_{XYZ}) \big]\\
        &\overset{\setlabel{step4}}{=} \lambda\,\,\underset{\Q_Z}{\inf} D_{\frac{1}{1-\lambda}}(\Pm_{XYZ}\|\Q_{Z}\Pm_{X|Z}\Pm_{Y|Z}) \\
        &\overset{\setlabel{step5}}{=} \lambda I_{\frac{1}{1-\lambda}}(X,Y|Z).
\end{align*}
Where step \reflabel{step1} follows from an analogous result of \cite[Corollary 2]{lapidothTesting} for our testing problem, step \reflabel{step2} follows because, given that $\lambda \leq 0$ then $D(\mathcal{R}_{XYZ}\|\Q_{Z}\Pm_{X|Z}\Pm_{Y|Z})$ achieves the maximum. Step \reflabel{step3} follows from an analogous of \cite[Lemma 4]{lapidothTesting}, \reflabel{step4} follows from \cite[Theorem 3]{RenyiKLDiv} and to conclude \reflabel{step5} follows from Definition \ref{conditionalSibs}.
For $\lambda>0$ the reasoning is identical to \cite[Lemma 12]{lapidothTesting}.
\end{proof}
Now, we can prove the connection to $E_P^{\star\star}(\cdot).$ \footnote{Notice that $E_P^{\star\star}(\cdot)$ is not guaranteed to be equal to $E_P(\cdot)$. Indeed, it is possible to find examples where the function is not convex and thus, all we retrieve is a lower bound on $E_P$ \cite[Example 14]{lapidothTesting}.
}
\begin{theorem}\label{biconjucate}
Given $E_Q,E_P\in\mathbb{R}$ 
\begin{align}
    E_P^{\star\star}(E_Q) &= \sup_{\alpha \in (0,1]} \frac{1-\alpha}{\alpha} (I_{\alpha}(X,Y|Z))-E_Q), \\
    E_Q^{\star\star}(E_P) &= \sup_{\alpha \in (0,1]}  \left(I_{\alpha}(X,Y|Z)-\frac{\alpha}{1-\alpha}E_P\right). \label{thmBiconjPart2}
\end{align}
\end{theorem}
\begin{proof}
\begin{align}
     E_P^{\star\star}(E_Q) & = \sup_{\lambda \in \mathbb{R}} (\lambda E_Q - E_P^\star(\lambda))\\
     & \overset{\setlabel{step1Proof2}}{=}  \sup_{\lambda \leq 0} (\lambda E_Q - E_P^\star(\lambda))\\
     & \overset{\setlabel{step2Proof2}}{=}  \sup_{\lambda \leq 0} (\lambda E_Q - I_{\frac{1}{1-\lambda}}(X,Y|Z))\\
     & \overset{\setlabel{step3Proof2}}{=} \sup_{\alpha \in (0,1]} \frac{1-\alpha}{\alpha} (I_{\alpha}(X,Y|Z)-E_Q).
\end{align}
Where \reflabel{step1Proof2} follows from $E_P^\star(\lambda)=+\infty$ for $\lambda>0$, \reflabel{step2Proof2} follows from Lemma \ref{conjugateEP} and \reflabel{step3Proof2} by setting $\alpha = \frac{1}{1-\lambda}$. The proof of \eqref{thmBiconjPart2} follows from similar arguments.
\end{proof}
\section{Conclusions}
We have considered the problem of defining a conditional version of Sibson's $\alpha$-Mutual Information. Drawing inspiration from an equivalent formulation of $I_\alpha(X,Y)$ as $\min_{\Q_Y}D_\alpha(\Pm_{XY}\|\Pm_X\Q_Y)$ we saw how several of these propositions can be made for a $I_\alpha(X,Y|Z)$. Two have already been analysed in \cite{tomamichel}. We proposed here a general approach that allows to connect to each such measure: \begin{enumerate}
    \item a bound, allowing to approximate the probability $\Pm_{XYZ}(E)$ with the probabilty of $E$ under a product distribution induced by the Markov chain $X-Z-Y$;
    \item an operational meaning as the error exponent of a hypothesis testing problem where the alternative hypothesis is a markov-like distribution and some measures are allowed to vary. 
\end{enumerate}   
A simple relationship between the hypothesis testing problem and the information measure can already be found  using the bound described in 1), without requiring any extra machinery. To conclude,
the usefulness of a measure clearly comes from its applications and ease of computability. While the latter remains the same for all the possible conditional $I_\alpha$ the former can vary according to the definition. With this in mind, the various definitions are equally meaningful and it seems reasonable to use the conditional $I_\alpha$ that best suits the specific application at hand.
\section*{Acknowledgment} The work in this paper was supported in part by the Swiss National Science Foundation under Grants 169294 and 200364.
\bibliographystyle{IEEEtran}
\bibliography{sample}

\begin{thebibliography}{10}
\providecommand{\url}[1]{#1}
\csname url@samestyle\endcsname
\providecommand{\newblock}{\relax}
\providecommand{\bibinfo}[2]{#2}
\providecommand{\BIBentrySTDinterwordspacing}{\spaceskip=0pt\relax}
\providecommand{\BIBentryALTinterwordstretchfactor}{4}
\providecommand{\BIBentryALTinterwordspacing}{\spaceskip=\fontdimen2\font plus
\BIBentryALTinterwordstretchfactor\fontdimen3\font minus
  \fontdimen4\font\relax}
\providecommand{\BIBforeignlanguage}[2]{{%
\expandafter\ifx\csname l@#1\endcsname\relax
\typeout{** WARNING: IEEEtran.bst: No hyphenation pattern has been}%
\typeout{** loaded for the language `#1'. Using the pattern for}%
\typeout{** the default language instead.}%
\else
\language=\csname l@#1\endcsname
\fi
#2}}
\providecommand{\BIBdecl}{\relax}
\BIBdecl

\bibitem{fullVersionGeneralization}
\BIBentryALTinterwordspacing
A.~R. Esposito, M.~Gastpar, and I.~Issa, ``Generalization error bounds via
  r{\'{e}}nyi-, f-divergences and maximal leakage,'' \emph{Accepted for
  Publication in IEEE Transactions on Information Theory}, 2021. [Online].
  Available: \url{http://arxiv.org/abs/1912.01439}
\BIBentrySTDinterwordspacing

\bibitem{gallager1}
R.~Gallager, ``A simple derivation of the coding theorem and some
  applications,'' \emph{IEEE Transactions on Information Theory}, vol.~11,
  no.~1, pp. 3--18, 1965.

\bibitem{gallager2}
R.~G. Gallager, \emph{Information Theory and Reliable Communication}.\hskip 1em
  plus 0.5em minus 0.4em\relax USA: John Wiley \& Sons, Inc., 1968.

\bibitem{leakageLong}
I.~{Issa}, A.~B. {Wagner}, and S.~{Kamath}, ``An operational approach to
  information leakage,'' \emph{IEEE Transactions on Information Theory},
  vol.~66, no.~3, pp. 1625--1657, 2020.

\bibitem{tomamichel}
M.~{Tomamichel} and M.~{Hayashi}, ``Operational interpretation of {R}\'enyi
  information measures via composite hypothesis testing against product and
  markov distributions,'' \emph{IEEE Transactions on Information Theory},
  vol.~64, no.~2, pp. 1064--1082, 2018.

\bibitem{alphaLeakage}
J.~Liao, L.~Sankar, O.~Kosut, and F.~P. Calmon, ``Robustness of maximal
  $\alpha$-leakage to side information,'' in \emph{2019 IEEE International
  Symposium on Information Theory (ISIT)}, 2019, pp. 642--646.

\bibitem{verduAlpha}
S.~Verd{\'{u}}, ``{\(\alpha\)}-mutual information,'' in \emph{2015 Information
  Theory and Applications Workshop, {ITA} 2015, San Diego, CA, USA, February
  1-6, 2015}, 2015, pp. 1--6.

\bibitem{RenyiKLDiv}
T.~van Erven and P.~Harremo\"es, ``{R}\'enyi divergence and
  {K}ullback-{L}eibler divergence,'' \emph{IEEE Trans. Inf. Theory}, vol.~60,
  no.~7, pp. 3797--3820, July 2014.

\bibitem{opMeanRDiv1}
I.~{Csiszar}, ``{G}eneralized cutoff rates and {R}\'enyi's information
  measures,'' \emph{IEEE Transactions on Information Theory}, vol.~41, no.~1,
  pp. 26--34, Jan 1995.

\bibitem{infoRadius}
R.~Sibson, ``Information radius,'' \emph{Z. Wahrscheinlichkeitstheorie verw
  Gebiete 14}, pp. 149--160, 1969.

\bibitem{lapidothTesting}
A.~{Lapidoth} and C.~{Pfister}, ``{T}esting against independence and a
  {R}\'enyi information measure,'' in \emph{2018 IEEE Information Theory
  Workshop (ITW)}, 2018, pp. 1--5.

\bibitem{sdpiRaginsky}
M.~{Raginsky}, ``Strong data processing inequalities and $\phi $ -sobolev
  inequalities for discrete channels,'' \emph{IEEE Transactions on Information
  Theory}, vol.~62, no.~6, pp. 3355--3389, 2016.

\bibitem{distrFuncComput}
A.~{Xu} and M.~{Raginsky}, ``Information-theoretic lower bounds for distributed
  function computation,'' \emph{IEEE Transactions on Information Theory},
  vol.~63, no.~4, pp. 2314--2337, 2017.

\bibitem{lectureNotesSDPI}
Y.~Wu, ``Lecture notes on: Information-theoretic methods for high-dimensional
  statistics,'' 2020.

\bibitem{bayesRiskRaginsky}
A.~{Xu} and M.~{Raginsky}, ``Information-theoretic lower bounds on bayes risk
  in decentralized estimation,'' \emph{IEEE Transactions on Information
  Theory}, vol.~63, no.~3, pp. 1580--1600, 2017.

\bibitem{lapidothTestingJournal}
\BIBentryALTinterwordspacing
A.~Lapidoth and C.~Pfister, ``Two measures of dependence,'' \emph{Entropy},
  vol.~21, no.~8, 2019. [Online]. Available:
  \url{https://www.mdpi.com/1099-4300/21/8/778}
\BIBentrySTDinterwordspacing

\end{thebibliography}
\end{document}